\documentclass[conference,a4paper]{IEEEtran}
\pdfoutput=1

\usepackage{pgfplots}
\pgfplotsset{compat=newest}
\usepackage{tikz}
\usetikzlibrary{arrows,matrix,positioning,patterns}
\usepackage[utf8]{inputenc}
\usepackage[innermargin=0.545in,outermargin=0.65in,top=0.545in,bottom=.7in]{geometry}
\usepackage[english]{babel}
\usepackage[T1]{fontenc}
\usepackage{epsfig}
\usepackage{amsmath, amssymb, amsbsy}
\usepackage{mathdots}
\usepackage{xspace}
\usepackage[noend]{algpseudocode}
\usepackage[linesnumbered,ruled,vlined,titlenumbered]{algorithm2e}
\usepackage{algorithmicx}
\usepackage{color}
\usepackage{cite}
\usepackage{booktabs}
\usepackage{verbatim}
\usepackage{url}
\usepackage{lipsum}
\usepackage{enumitem}
\usepackage{colortbl}
\usepackage{amsthm}
\usepackage{dblfloatfix}
\usepackage{verbatim}
\makeatletter
\newcommand\footnoteref[1]{\protected@xdef\@thefnmark{\ref{#1}}\@footnotemark}
\makeatother

\newtheorem{theorem}{Theorem}

\newtheorem{remark}[theorem]{Remark}

\usepackage[final,tracking=true,kerning=true,spacing=true,factor=1100,stretch=10,shrink=20]{microtype}

\newenvironment{mymatrix}{\begin{bmatrix}} {\end{bmatrix} }

\def\ve#1{{\mathchoice{\mbox{\boldmath$\displaystyle #1$}}%
              {\mbox{\boldmath$\textstyle #1$}}%
              {\mbox{\boldmath$\scriptstyle #1$}}%
              {\mbox{\boldmath$\scriptscriptstyle #1$}}}}

\renewcommand{\vec}[1]{\ensuremath{\boldsymbol{#1}}}

\newcommand{\Fq}{\ensuremath{\mathbb{F}_q}}
\newcommand{\Fqm}{\ensuremath{\mathbb{F}_{q^m}}}

\newcommand{\mycode}[1]{\ensuremath{\mathcal{#1}}}

\newcommand{\y}{\ve{y}}

\newcommand{\e}{\ve{e}}

\newcommand{\m}{\ve{m}}
\renewcommand{\c}{\ve{c}}
\renewcommand{\e}{\ve{e}}
\renewcommand{\m}{\ve{m}}
\newcommand{\G}{\ve{G}}

\newcommand{\F}{\mathbb{F}}

\usepackage{graphics} 
\usepackage{epsfig} 

\usepackage{times} 
\usepackage{amsmath} 
\usepackage{amssymb} 
\usepackage{cite}
\usepackage{soul}
\usepackage[mathscr]{eucal}

\newtheorem{cor}{Corollary}

\newtheorem{const}[cor]{Construction}
\newtheorem{thmmystyle}{Theorem}

\usepackage{tikz}
\usetikzlibrary{shapes}
\usetikzlibrary{shapes.multipart,chains}
\makeatletter
\newcommand{\removelatexerror}{\let\@latex@error\@gobble}
\makeatother
\makeatletter
\newcommand*{\rom}[1]{\expandafter\@slowromancap\romannumeral #1@}
\makeatother

\IEEEoverridecommandlockouts

\begin{document}

\title{Error Correction for \\ Partially Stuck Memory Cells}
\author{\IEEEauthorblockN{Haider Al Kim$^{1,2}$\thanks{This work has received funding from the German Research Foundation (Deutsche Forschungsgemeinschaft, DFG) under Grant No. WA3907/1-1 and from the German Israeli Project Cooperation (DIP) grant no.~KR3517/9-1. Al Kim has received funding from the German Academic Exchange Service (Deutscher Akademischer Austauschdienst, DAAD) under the support program ID 57381412.}, Sven Puchinger$^{1}$, Antonia Wachter-Zeh$^{1}$}

\IEEEauthorblockA{
  $^1$Institute for Communications Engineering, Technical University of Munich (TUM), Germany\\ $^2$Electrical and Communication Engineering, University of Kufa, Iraq\\
  Email: \{haider.alkim, sven.puchinger, antonia.wachter-zeh\}@tum.de}

  }

\maketitle

\begin{abstract}
	We present code constructions for masking $u$ partially stuck memory cells with $q$ levels and correcting additional random errors. The results are achieved by combining the methods for masking and error correction for stuck cells in \cite{heegard1983partitioned} with the masking-only results for partially stuck cells in \cite{wachterzeh2016codes}. We present two constructions for masking $u<q$ cells and error correction: one is general and based on a generator matrix of a specific form. The second construction uses cyclic codes and allows to efficiently bound the error-correction capability using the BCH bound. Furthermore, we extend the results to masking $u\geq q$ cells. For $u>1$ and $q>2$, all new constructions require less redundancy for masking partially stuck cells than previous work on stuck cells, which in turn can result in higher code rates at the same masking and error correction capability.
  
\end{abstract}


\section{Introduction}
The dominance of non-volatile memories such as PCMs (phase change memories) as memory solutions for a variety of applications has become significant due to their advantages as permanent storage devices. The advantages of these memories are their rapid increase in capacity plus their ability as multi-levels technologies. For these reasons, their cost has been reduced strongly in the last years. However, reliability issues make it necessary to suggest new sophisticated coding and signal processing solutions. PCM cells can hold two states: an amorphous state and several crystalline states. 
Non-defective memory cells can switch between their main states (amorphous and crystalline). However, due to the cooling and heating processes of the cells, PCMs may face failures in changing their states, and therefore the cells can hold only one phase and they become \emph{stuck} \cite{gleixner2009reliability,kim2005reliability,lee2009study,pirovano2004reliability}.

This means that the cell’s charge is trapped in the cell, and it cannot change its status to be re-written. 
To deal with the stuck positions, a mechanism called \emph{masking} is used. Masking finds a codeword that holds the same levels as in the stuck positions, and can therefore be placed properly on the memory.
For multi-level PCMs, since the crystalline state can be programmed into partial states, the cell may be stuck in this level or in one of its sub-levels (level higher than $0$).
If the cell can only represent levels greater or equal to a reference level $s>0$, it is called a \emph{partially stuck cell} \cite{wachterzeh2016codes}. For multi-level PCMs, the case $s=1$ is particularly important since this means that a cell cannot reach the amorphous state anymore, but all partially crystalline ones, cf.~\cite{wachterzeh2016codes}.
Similarly, (partially) stuck cells can occur in flash memories. Flash memory stores information by charging it electronically. If charge is trapped inside one cell at a certain level (one of the main levels or any intermediate levels), the ways to rewrite on again are by increasing the trapped level or by erasing one whole block. However, erasing whole block reduces the lifetime of these memory devices. 
Figure~1 shows the general idea of the (partially) stuck memory cells.
On the other hand, it happens that due to manufacturing defects, cells can only hold lower levels causing the reverse problem as partially stuck cells.

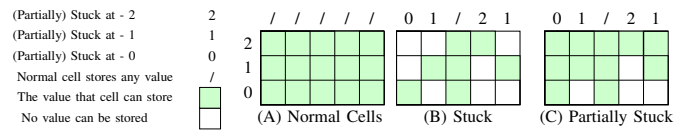
\begin{figure} [h]
	\scalebox{0.55}{	
		\begin{tikzpicture}
		\draw
		
		(3.25,0) node[anchor=north] {\small {No value can be stored}} 
		
		;
		
		\draw (6,- 0.5) rectangle (6.5,0);
		\draw(3.5,0.5) node[anchor=north] {\small {The value that cell can store}} 
		;

		\fill[green!20!white] (6,0) rectangle (6.5,0.5);
		\draw (6,0) rectangle (6.5,0.5);
		\draw (6.3,1) node[anchor=north] {/}
		
		(3.5,1) node[anchor=north] {\small {Normal cell stores any value}}
		;

		\draw (6.3,1.5) node[anchor=north] {0}
		
		(3,1.5) node[anchor=north] {\small{(Partially) Stuck at - 0}}
		;
		\draw (6.3,2) node[anchor=north] {1}
		
		(3,2) node[anchor=north] {\small{(Partially) Stuck at - 1}}
		;
		\draw (6.3,2.5) node[anchor=north] {2}
		(3,2.5) node[anchor=north] {\small{(Partially) Stuck at - 2}}
		;
		\end{tikzpicture}
	}
	\scalebox{0.65}{	
		
		\begin{tikzpicture}[baseline={(0,-0.5)}]
		
		\draw
		(-0.25,1.5) node[anchor=north] {2}
		(-0.25,1) node[anchor=north] {1}
		(-0.25,0.5) node[anchor=north] {0}
		(0.25,2) node[anchor=north] {/}
		(0.75,2) node[anchor=north] {/}
		(1.25,2) node[anchor=north] {/}
		(1.75,2) node[anchor=north] {/}
		(2.25,2) node[anchor=north] {/}
		
		;
		\draw[thick,-] (0,1.5) -- (2.5,1.5);
		
		\draw (0,0) rectangle (2.5,1.5);
		
		\fill[green!20!white] (0,0) rectangle (2.5,1.5);
		\draw (0,0) rectangle (2.5,0.5);
		\draw (0,0) rectangle (2.5,1);
		\draw (0,0) rectangle (0.5,1.5);
		\draw (0,0) rectangle (1,1.5);
		\draw (0,0) rectangle (1.5,1.5);
		\draw (0,0) rectangle (2,1.5);
		\draw[thick,-] (2,1.5) -- (2.5,1.5);
		\draw[thick,-] (2.5,1) -- (2.5,1.5);
		
		\draw
		(3,2) node[anchor=north] {0}
		(3.5,2) node[anchor=north] {1}
		(4,2) node[anchor=north] {/}
		(4.5,2) node[anchor=north] {2}
		(5,2) node[anchor=north] {1};
		\draw (2.75,0) rectangle (5.25,1.5);
		\fill[green!20!white] (2.75,0) rectangle (3.25,0.5);
		\fill[green!20!white] (3.25,0.5) rectangle (3.75,1);
		\fill[green!20!white] (3.75,0) rectangle (4.25,1.5);
		\fill[green!20!white] (4.25,1) rectangle (4.75,1.5);
		\fill[green!20!white] (4.75,0.5) rectangle (5.25,1);
		\draw (2.75,0) rectangle (5.25,0.5);
		\draw (2.75,0) rectangle (5.25,1);
		\draw (2.75,0) rectangle (3.25,1.5);
		\draw (2.75,0) rectangle (3.75,1.5);
		\draw (2.75,0) rectangle (4.25,1.5);
		\draw (2.75,0) rectangle (4.75,1.5);
		;
	
		\draw
		
		(6,2) node[anchor=north] {0}
		(6.5,2) node[anchor=north] {1}
		(7,2) node[anchor=north] {/}
		(7.5,2) node[anchor=north] {2}
		(8,2) node[anchor=north] {1};
		\draw[thick,-] (6,1.5) -- (6,1.5);
		\draw (5.75,0) rectangle (7.75,1.5);
		\fill[green!20!white] (5.75,0) rectangle (6.25,0.5);
		\fill[green!20!white] (6.25,0.5) rectangle (6.75,1);
		\fill[green!20!white] (6.75,0) rectangle (7.25,1.5);
		\fill[green!20!white] (5.75,1) rectangle (8.25,1.5);
		\fill[green!20!white] (5.75,0.5) rectangle (6.25,1);
		\draw[thick,-] (7.75,1) -- (7.75,1.5);
		\fill[green!20!white] (7.75,0.5) rectangle (8.25,1);
		\draw (5.75,0) rectangle (8.25,0.5);
		\draw (5.75,0) rectangle (8.25,1);
		\draw (5.75,0) rectangle (6.25,1.5);
		\draw (5.75,0) rectangle (6.75,1.5);
		\draw (5.75,0) rectangle (7.25,1.5);
		\draw (5.75,0) rectangle (8.25,1.5);
		\draw[thick,-] (7.75,1.5) -- (8.25,1.5);
		\draw[thick,-] (7.75,.5) -- (7.75,1);
		
		\draw
		(1.2,0) node[anchor=north] {(A) Normal Cells} ;
		\draw (4,0) node[anchor=north] {(B) Stuck} 
		(7,0) node[anchor=north] {(C) Partially Stuck};
		\end{tikzpicture}
	}
	
	\caption{General idea of (partially) stuck memory cells. In this figure, there are $n=5$ cells with $q=3$ possible levels. The stuck levels are (0, 1 or 2). In case (A), normal cells can store any of the three values. In the stuck scenario as shown in case (B), the stuck cell can store only the exact stuck level $s$. It is more flexible in case (C) (partially stuck scenario). A cell that is partially stuck at 0 can store any value as a non-defective cell. Partially stuck cells at level $s \geq 1$ can store one level or more. This paper only deals with partially stuck at 1 cells.}	

\end{figure}
\vspace{-0.3cm}
\subsection{Related Work}

In \cite{heegard1983partitioned}, code constructions for masking \emph{stuck memory cells} were proposed. In addition to masking the defect cells, it is possible to correct errors that occur during the storing and reading processes. A generator matrix of a specific form was constructed for this purpose. Moreover, in \cite{heegard1983partitioned}, a partitioned cyclic code and partitioned BCH were proposed to mask stuck cells and correct errors. However, the required redundancy for masking more than one cell is greater than one symbol. Later, an asymptotic optimal analysis for $n \rightarrow \infty$ was proposed to mask defects of (fixed\footnote{An asymptotically optimal class of codes was proposed to correct $\rho$ defects ($\rho = \text{const value}$) regardless of $n \rightarrow \infty$, where $n$ is the code length.} \cite{dumer1989asymptotically}, linearly increasing\footnote{Asymptotically optimal linear codes were proposed to correct $\rho$ defects $(\rho = \alpha n)\text{, where } \alpha = \text{ constant value}$ and $\rho =$ number of defects that is linearly increasing while $n \rightarrow \infty$, i.e. if $n$ is doubled then $\rho$ is doubled.} \cite{dumer1990asymptotically}) multiplicity that need a redundancy of at least the number of defects. Besides masking the defects, the construction in \cite [Section 5]{dumer1990asymptotically} can correct errors with overall redundancy $r(n,\rho,\tau) = \rho + o(n)$, where $\rho$ denotes the number of defects and $\tau$ denotes the number of errors. $o(n)$ is the required redundancy to correct errors. 

In \cite{wachterzeh2016codes}, improvements on the redundancy necessary for masking \emph{partially stuck memory cells} are achieved, and lower and upper bounds are derived by the constructions.
However, the paper does not consider error correction in addition to masking.

\subsection{Our Contribution}

In this paper, we combine the methods of \cite{heegard1983partitioned} and \cite{wachterzeh2016codes}. We obtain code constructions for combined error correction and masking partially stuck cells.
Compared to the stuck-cell case in \cite{heegard1983partitioned}, we reduce the redundancy necessary for masking, similar to the results in \cite{wachterzeh2016codes}.\\
In contrast to \cite{wachterzeh2016codes}, however, we are able to correct additional random errors.
Similar to the main part of \cite{wachterzeh2016codes}, this paper deals with partially-stuck-at-1 cells, i.e., $s=1$ (recall that this is the typical for PCMs, cf.~\cite{wachterzeh2016codes}), but the results are extendable to arbitrary $s$ similar to \cite[Section~VII]{wachterzeh2016codes}.\\
On the other hand, compared to the stuck cell model in \cite [Section 5]{dumer1990asymptotically}, it is not clear if $r(n,\rho,\tau) = \rho + o(n)$ is higher or lower than our overall redundancy shown in Theorem~\ref{thm4}. This is because we do not consider an asymptotic analysis in this paper for the \emph{partially stuck} cells with errors correction model. However, under the assumption that both of them
require $o(n)$ redundancy to correct errors, our overall redundancy shown in Theorem~\ref{thm4} is lower. The reason is that \cite [Section 5]{dumer1990asymptotically} uses $\rho$ check symbols to mask the \emph{stuck} cells while our construction uses less than $\rho$ as a redundancy necessary to mask \emph{partially stuck} cells. Moreover, our constructions can correct a certain number of errors and mask a certain number of \emph{partially stuck} cells, while \cite[Section 5]{dumer1990asymptotically} proposed $o(n)$ redundancy for error correction that is negligible, i.e $o(n) \rightarrow 0 \iff n \rightarrow \infty$.

\section{Preliminaries}

\subsection{Notations }\label{ssec:notation}
For a prime power $q$, let $\mathbb{F}_q$ denote the finite field of order $q$ and ${\mathbb{F}}_q[x]$ be the set of all univariate polynomial with coefficients in $\mathbb{F}_q$. Write $[f] = \{0,1, \dots, f-1\}$. Let $k_1$ be the number of information symbols, $l$ be the required symbol(s) for masking, $r$ be the required redundancy for error correction, $t$ be the number of errors, $u$ be the number of (partially) stuck cells.
Let $s_{\phi_i}$ denote the (partially) stuck level at any position, where $i \in [u]$. Let $n$ be the code length and also the memory size.  
\subsection{Definitions}
\subsubsection{Stuck and Partially Stuck Cells}
A cell is called \emph{stuck at level $s$}, where $ s \in [q]$, if it cannot change its value and always stores the value $s$. 
A cell is called \emph{partially stuck at level $s$}, where $s \in [q]$, if it can only store values which are at least $s$. 
If a cell is partially stuck at 0, it is a non-defect cell which can store any of the $q$ levels  \cite {wachterzeh2016codes}.
\subsubsection{($u$, $t$)-PSMC}
An $(n,M)_q$  ($u$, $t$)-\emph{partially-stuck-at-masking code} $\mycode{C}$ is a coding scheme consisting of an encoder $\mathcal{E}$ and decoder $\mathcal{D}$.
The input of the encoder $\mathcal{E}$ is
\begin{itemize}
\item the set of locations of $u$ partially stuck cells $\ve{\phi}= \{\phi_0,\phi_1,\phi_2, \cdots,\phi_{u-1}\} \subseteq [n]$, 
\item the partially stuck levels $s_{\phi_0},s_{\phi_2}, \cdots,s_{\phi_{u-1}}  \in [q]$, and
\item a message $\m \in \mathcal{M}$, where $\mathcal{M}$ is a message space of cardinality $|\mathcal{M}| = M$
\end{itemize}
It outputs a vector $\c \in \Fq^n$ which fulfills ${c}_{\phi_i} \geq s_{\phi_i}$ for all $i=1,\dots,u$.
The decoder is a mapping that takes $\c+\e \in [q]^n$ as input and returns the correct message $\m$ for all error vectors $\e$ of Hamming weight at most $t$.
\subsubsection{$Q$-ary Cyclic Code} a $q$-ary cyclic code of length $n$, dimension $k$, and minimum distance $d$ is denoted as an $[n,k,d]_q$ code $\mathcal{C}$. It has a generator polynomial $g(x)$ of degree $n-k$ with roots in $\mathbb{F}_{q^m}$, where $n$ divides $q^m - 1$.

A cyclotomic coset $M_a$ is given by:
\begin{equation}
M_a := \left\{a \cdot q^j \mod n ,\text{ } \forall j = 0,1,\cdots, n_a-1 \right\},
\end{equation}
where $n_a$ is the smallest integer such that $a\cdot q^{n_a} \equiv a \mod n $.
Let $\alpha \in \Fqm$ be a primitive $n^{th}$ root of unity in $\mathbb{F}_{q^m}$.
The minimal polynomial of an element $\alpha^a$ is given by:
\begin{equation}
M^{(a)} (x) :=  \prod_{b \in M_a} (x-\alpha^b)
\end{equation}
Although the factors $(x-\alpha^b)$ are in $\Fqm[x]$, minimal polynomials have coefficients in the small field $\Fq$, i.e., $M^{(a)} (x) \in \mathbb{F}_q[x]$. 
The defining set $D_c$ of a $q$-ary cyclic code $\mathcal{C}$ with parameters $[n,k,d_1]_q$ is the set containing the indices $b$ of the zeros $\alpha^b$ of the generator polynomial $g(x)$. If $a \in D_c$, then we have $M_s \subseteq D_c$, and hence, $D_c$ is a union of cyclotomic cosets $M_{a_1},\dots,M_{a_w}$ for some $w$, i.e.
\begin{equation}
D_c := \{ b: g(\alpha^b)= 0\} = M_{a_1} \cup M_{a_2} \cup M_{a_3} \cdots \cup M_{a_w}.
\end{equation} 
The generator polynomial $g(x) \in \mathbb{F}_q[x]$ of the code $[n,k,d_1]$ of degree $r= n-k$ is thus given by
\begin{equation}
g(x) = \prod_{a \in D_c} (x-\alpha^a) = \prod_{b =1}^{w} M^{(a_b)} (x).
\end{equation} 
For any cyclic code, there is a parity-check polynomial:
\begin{equation}
h(x) = \dfrac{(x^n -1)}{g(x)} = \prod_{a \in [n] \setminus D_c} (x-\alpha^a).
\end{equation}
The minimum distance of a cyclic code is at least its BCH bound, which is the number of consecutive elements in $D_c$ plus one.

\section{Codes for (Partially) Stuck Cells }   

\subsection{Masking Partially Stuck and Correcting Additional Random Errors using a Generator Matrix Construction}

The goal of this paper is to find a code construction that can store information in a memory with some partially stuck cells and additionally can correct errors. For the masking process and according to \cite{wachterzeh2016codes}, we need only a single redundancy symbol if $u<q$ and ($s_{\phi_i} =1$). However, the work in \cite{wachterzeh2016codes} does not consider correcting additional random errors. The following theorem introduces a code construction using a generator matrix with a specific form that masks partially stuck cells and corrects errors.

\begin{thmmystyle}\label{thm:matrix_construction:u<q}
	Let $u\leq \min \{n,q-1\}$. Assume there is an $[n, k, d \geq 2t+1]_q$ code $\mycode{C}$ with  a $k \times n$ generator matrix of the following form:
	
	\begin{center}			
	
	$ \ve{G} = \begin{bmatrix} \ve{G}_1 \\ \ve{G}_0  \end{bmatrix} = \begin{bmatrix} \ve{0}_{k_1 \times 1} & \ve{I}_{k_1} & \ve{P}_{k_1\times r} \\  {1} & {\ve{1}_{1 \times k_1}} & {\ve{1}_{1 \times r}} \end{bmatrix}$,
	\end{center}
	where: 
	\begin{itemize}
		\item $\ve{G}_1$ is a $k_1 \times n$ generator matrix of an $[n, k_1]_q$ code $\mycode{C}_1$. 
		\item $k_1 = n - 1 - r$.
		\item $k = k_1 + 1$.
		\item $\ve{P} \in \mathbb{F}^{k_1\times r}_q $.
	\end{itemize}
	
	By using Algorithm~\ref{tab:a1} and ~\ref{tab:a2}, this code is a ($u$,$t$)-PSMC which can mask $u$ partially stuck memory cells, and can correct $t$ additional random errors while storing $k_1 = n-r-1$ information symbols. 	

     \label{thm1}
\end{thmmystyle}
\begin{algorithm}
	\caption{Encoding}
	\label{tab:a1}
	
	\KwIn{
	\begin{itemize}
		\item Message: 
		$\m = (m_0, m_1,\dots, m_{k_1-1})  \in {\mathbb{F}_q^{k_1}}$
		\item Positions of partially stuck cells:
		 $\ve{\phi}$
	\end{itemize}}
	
	Compute the message vector $\ve{w} = \ve{m} \cdot \ve{G}_1$

	Find $v \in \Fq$ such that $w_{\phi_i} \neq v$ 

	$z_0 \gets q-v \equiv -v \mod q$

	Compute the masking vector $\ve{d} = z_0 \cdot \ve{G_0}$

	Compute $\ve{c} = (\ve{w} + \ve{d})  \mod q$	

	\KwOut{Codeword $\ve{c} \in \mathbb{F}^n_q$ }
	
\end{algorithm}
\begin{algorithm}
		\label{tab:a2}
	\caption{Decoding}	
		\KwIn{
	\begin{itemize}
	     \item Retrieve $\ve{y} = \ve{c} + \ve{e}$ , $\ve{y} \in \mathbb{F}^n_q $
        \end{itemize} }
 
	$\hat{\c} \gets$ decode $\ve{y}$ in $\mathcal{C}$

	$\hat{z}_0 \gets$ first entry of $\c$

	$\hat {\ve{w}} = (\hat{w}_0, \hat{w}_1,\cdots, \hat{w}_{n-1}) \leftarrow (\hat{\ve{c}} - \hat{z}_0 \cdot \ve{G}_0)  \mod q$

	$\hat{\ve{m}}  \leftarrow (\hat{w}_1, \dots, \hat{w}_{k_1})$

	\KwOut{
	 Message  vector $\hat{\ve{m}}  \in \mathbb{F}^{k_1}_q$}  

\end{algorithm}	
\begin{proof}	 
	Since in partially stuck cells the stuck level $s_{\phi_i} \geq 1$ has to be masked, the output codeword $\ve{c}$ must match the partially stuck positions:
	\begin{equation} \label{equ}
	c_{\phi_0},c_{\phi_1},\dots, c_{\phi_{u-1}} \geq 1.
	\end{equation}
	Since $u < q$, there is at least one value $v \in \mathbb{F}_q$ such that ${w_\phi}_0, {w_\phi}_1,  \dots , {w_\phi}_{u-1} \not = v$.
	Thus, we choose $z_0 = q-v \equiv -v \mod q$. Therefore ${c}_{\phi_i}=w_{\phi_i}+z_0 \equiv ({w_\phi}_i - v) \mod q \not = 0$ and \eqref{equ} is satisfied.
	
	The decoder (Algorithm~\ref{tab:a2}) gets $\ve{y}$, which is $\c$ corrupted by at most $t$ errors. Since $\mathcal{C}$ has minimum distance $d \geq 2 t+1$, we can correct these errors and obtain $\c$.
	Due to the structure of the matrix $\G$, the first position of $\c$ equals the masking value $z_0$. Hence, we can compute $\hat{\ve{w}} = \ve{w}$ (cf.~Algorithm~\ref{tab:a2}) and the message vector $\hat {\mathbf{m}}=\mathbf{m}$. 
\end{proof}
	 Theorem~\ref{thm1} gives an extension of \cite[Theorem~1]{heegard1983partitioned} and \cite[Theorem~4]{wachterzeh2016codes}. It combines \cite{heegard1983partitioned} and \cite{wachterzeh2016codes} to provide a code construction that can mask partially stuck cells and correct errors. The required redundancy is a single symbol for masking plus the redundancy for the code $\mathcal{C}_1$. In comparison, \cite[Theorem~1]{heegard1983partitioned} requires at least 
\begin{align*}
\min\{ n-k \, : \, \exists \, [n,k,d]_q \text{ code with } d > u\} \geq u.
\end{align*}
redundancy symbols to mask $u$ cells, where the inequality follows directly from the Singleton bound.

The code construction of Theorem~\ref{thm1} is based on the matrix $\ve{P}$, which is hard to find in general. One way to construct such a matrix is to start with a generator matrix of a well-known code that contains the all-one vector (e.g. certain cyclic codes, see the next sections) and has a large enough minimum distance. It is easy to see that the generator matrix can be transformed into the form of Theorem~\ref{thm:matrix_construction:u<q} by elementary row operations and column permutations. 

\subsection{Masking Partially Stuck and Correcting Additional Random Errors using a Partitioned Cyclic Code Construction }

This section generalizes the construction of \cite[Theorem~2]{heegard1983partitioned}. Our construction uses a so-called \emph{partitioned cyclic code} as in \cite{heegard1983partitioned}, but requires only a single redundancy symbol $l=1$ for the masking operation similar to Theorem 4 and Algorithm 3 in \cite{wachterzeh2016codes}. 
Additionally, it corrects errors by selecting a generator polynomial $g_1(x)$ of degree $r$, where $0 \leq r < n-1$.
As the construction in the previous section, this new construction results in a reduced redundancy $l$ compared to masking stuck cells in \cite[Theorem~2]{heegard1983partitioned}.
Compared to Theorem~\ref{thm:matrix_construction:u<q} in the previous section, the cyclic construction directly implies a constructive strategy how to choose a code of a certain minimum distance.
\begin{const}\label{cons1}
	Let $u\leq \min \{n,q-1\}$. Let $\mycode{C}_0$ be an $[n, n-1,\delta_0]_q$ code with parity-check polynomial $g_0(x) = 1+x+x^2+\dots+x^{n-1}$.
	Let:
	\begin{equation*}
	c_0(x) = z_0 \cdot g_0(x).
	\end{equation*}
	Let $\mycode{C}_1$ be an $[n,n-r,\delta_1]_q$ code 
	where $g_1(x)$ is its monic generator polynomial of degree $r$ that divides $g_0(x)$, and $h_1(x)$ is its parity-check polynomial.
	Then let $\mycode{C}$ be an $[n,n-r-1,\delta_1,\delta_0]_q$ partitioned cyclic code built from $\mycode{C}_1$ and $\mycode{C}_0$ with the following encoding rule:
	\begin{equation*}
	c(x) =  m(x) \cdot g_1(x) + z_0 \cdot g_0(x)  , 
	\end{equation*}
	where $m(x)\in \mathbb{F}_q[x]$ is the message polynomial of degree $< n-r-1$ and $z_0$ is a scalar.
\end{const}

\begin{table*}[b]
	\caption{ Ternary Codes for Partially Stuck-at-1 Memory Cell}
	\label{table2}
	\vspace{-.5 cm}
	\begin{center}
		\scalebox{0.9}{
			\begin{tabular}{ |l| l | l | l |l |l |l| l | l |l| l | l |}
				\hline
				\#&$k_1$ &$k_1^*$ &$l$ &$l^*$ &$r$&$\delta_0$ &$\delta_1$ &$t$&$h_0(x)$&$g_1(x)$&\textbf{Comment} \\ \hline
				1&6 &6.387& 1& 0.613& 1 & 2& 2&0&$M^{(4)}(x)$&$M^{(0)}(x)$& This is masking only with one redundancy symbol. \\ \hline
				2&5&5.387& 1 &0.613& 2 & 2 &2 &0 &$M^{(4)}(x)$&$M^{(5)}(x)$& Flexibility in storing a desired $k$ information symbols. \\ \hline
				3&4&4.387& 1& 0.613&  3& 2 & 3&1 &$M^{(4)}(x)$&$M^{(0)}(x)\cdot M^{(1)}(x)$& Same as Table~\ref{table3} in assigning factors $\rightarrow$ same parameters.\\ \hline
				
				4&3&3.387& 1&0.613& 4 &2  & 3 & 1&$M^{(4)}(x)$&$M^{(1)}(x)\cdot M^{(2)}(x)$&Same as Table~\ref{table3} in assigning factors $\rightarrow$ same parameters.\\ \hline
				5&2&2.387&1 & 0.613&5  &2  &5 &2 &$M^{(4)}(x) $&$M^{(0)}(x)\cdot M^{(1)}(x)\cdot M^{(2)}(x)$& Same as Table~\ref{table3} in assigning factors $\rightarrow$ same parameters.   \\ \hline
				6&2&2.387&1&0.613&5&2&3&1&$M^{(4)}(x)$&$M^{(0)}(x)\cdot M^{(1)}(x)\cdot M^{(5)}(x)$& No change in $\delta_1$, however it is more likely to be increased. \\ \hline
				
				7& 1&1.387& 1&0.613 & 6 & 2  &4 &1 &$M^{(4)}(x)$&$M^{(1)}(x)\cdot M^{(2)}(x)\cdot M^{(5)}(x)$&Same as Table~\ref{table3}, but it tends to be able to correct more errors.\\ \hline

			\end{tabular}
		}
	\end{center}
\end{table*}
\begin{table*}[b]
	\caption{Ternary Codes for Stuck-at Memory \cite {heegard1983partitioned}}
	\label{table3}
	\vspace{-.6 cm}
	\begin{center}
		\scalebox{0.95}{
			\begin{tabular}{ |l| l | l | l |l |l |l| l | l |l|}
				\hline
				\#&$k_1$ & $l$ & $r$& $\delta_0$ &$\delta_1$ &$t$ &$h_0(x)$&$g_1(x)$& \textbf{Comment} \\ \hline
				1&6 & 1 & 1 & 2& 2&0&$M^{(4)}(x)$&$M^{(0)}(x)$& Masking only. \\ \hline
				2&4& 1 &  3& 2 & 3&1 &$M^{(4)}(x)$&$M^{(0)}(x)\cdot M^{(1)}(x)$&  Same as Table~\ref{table2} in assigning factors $\rightarrow$ same parameters.\\ \hline
				3&3& 1& 4 &2  & 4 & 1&$M^{(4)}(x)$&$M^{(1)}(x)\cdot M^{(2)}(x)$& Same as Table~\ref{table2} in assigning factors $\rightarrow$ same parameters.\\ \hline
				4&2& 1 & 5 & 2 &5 &2 &$M^{(4)}(x)$&$M^{(0)}(x)\cdot M^{(1)}(x)\cdot M^{(2)}(x)$&Same as Table~\ref{table2} in assigning factors $\rightarrow$ same parameters. \\ \hline
				5&2&3  &3  &3  &3 &1 &$M^{(4)}(x) \cdot M^{(5)}(x)$&$M^{(0)}(x)\cdot M^{(1)}(x)$& $l = 3$ less flexibility in chosen $r$.   \\ \hline
				6 & 1& 3 & 4 &3  &4 &1 &$M^{(4)}(x)\cdot M^{(5)}(x)$&$M^{(1)}(x)\cdot M^{(2)}(x)$& Same as Table~\ref{table2}. However, less likely to correct more errors.\\ \hline
				
			\end{tabular}
		}	
	\end{center}
\end{table*}

\begin{thmmystyle}\label{thm:cyclic_construction}
	If $u\leq \min \{n,q-1\}$, Construction~\ref{cons1} provides an $(n, M = q^{n-r-1})_q$ ($u$,$t$)-PSMC with redundancy of $1+r$ symbols by using Algorithms~\ref{a3} and \ref{a4}.
	\label{thm2}
\end{thmmystyle}
\vspace{-0.3cm}
\begin{algorithm} \label{a3}
	\caption{Encoding}
	
		\KwIn{
	\begin{itemize}
		\item Message:
		$m(x)  \in F_q[x]$ of degree $ < n-r-1$
		\item Positions of partially stuck cells:
		$\ve{\phi}$
	\end{itemize}}

		$n = q^m - 1$ , $m$ is an integer
		
		Let $g_0(x) =1+x+x^2+ \dots +x^{n-1}$ and $h_0(x) =x-1$
		
		Select $g_1(x)\mid g_0(x)$ of degree $r$, where $0 \leq r < n-1$ 
		 
		$c_1(x) = {c_1}_0 +\dots+ {c_1}_{n-2}\cdot x^{n-2} \leftarrow m(x) \cdot g_1(x)$
		 	
		Assign $v \in \mathbb{F}_q$ such that ${c_1}_{\phi_0}, {c_1}_{\phi_1}, \dots , {c_1}_{\phi_{u-1}} \not = v$
		
		Find $z_0 = -v \mod q$
		
		$c_0(x) = z_0 \cdot g_0(x)$
		 	
		$  c(x) = c_1(x) + c_0(x) \mod (x^n-1)$

    \KwOut{
	Codeword $c(x) \in \mathbb{F}_q[x]$ of degree $\leq n-1$}	
\end{algorithm}
\vspace{-0.4cm}
 \begin{algorithm}\label{a4}
	\caption{Decoding}	
	
		\KwIn{Retrieve $y(x) = c(x) + e (x)$}

	$\hat{c}(x) \gets$ Decode $y(x)$ in the code generated by $g_1(x)$.
		 
	$\hat{m}(x) \gets \hat{c}(x) \mod g_0(x)$
		
	\KwOut{Message $\hat m(x) \in \mathbb{F}_q[x]$ of degree $< n-r-1$ } 
	
\end{algorithm}
\vspace{-0.3cm}	
\begin{proof}	
	Algorithm~\ref{a3} shows the encoding process for the partitioning cyclic code construction. 
	Let $m(x) = m_0+m_1\cdot x+m_2\cdot x^2+\dots+ m_{n-r-2}\cdot x^{n-r-2}$. Algorithm~\ref{a3} calculates  ${c}_1(x)$ in Step~4 of degree $< n-1$.
	Since $u < q$, there is at least a single value $v \in \mathbb{F}_q$ such that the coefficients of $c_1(x)$ in the partially stuck positions ${c_1}_{\phi_0}, {c_1}_{\phi_1}, \dots , {c_1}_{\phi_{u-1}} \not = v$.
	So we choose $z_0 = -v \mod q$ as shown in Step~5. 
	Note that the value of $z_0$ always appears in $ c(x)$. That is, the last memory location stores $z_0$. Since in partially stuck cells the stuck level $s_{\phi_i} \geq 1$, the coefficients of $c(x)$ must match the partially stuck positions as in (\ref{equ}). 
	Equation~(\ref{equ}) is also satisfied because ${c}_{\phi_i}= {c_1}_{\phi_i}+z_0 \equiv ({c_1}_{\phi_i} - v) \mod q \neq 0$. 
	
	Algorithm~\ref{a4} decodes the retrieved polynomial $y(x)$.
	First, decode $y(x)$ with the code generated by $g_1(x)$. We can correct a random error $e(x)$ (if any) by taking the module operation $mod\text{ }g_1(x)$. We must choose $\hat e(x) \in \mathbb{F}^n_q $ which minimizes $e(x)$ and satisfies $\hat e(x) \mod g_1(x) = e(x) \mod g_1(x)$. 
	
	Then, the algorithm performs the unmasking process to find $\hat m(x)$. If $\hat m(x) = m(x)$ of degree $ < n-r-1$, the decoding is successful. Taking$\mod g_0(x)$ for $\hat{c}(x)$ gives $\hat m(x) = m(x)$:
		  \begin{equation*}
	\begin{array}{l}
	 \hat{c}(x) = \hat m(x)\cdot  g_1(x) + z_0 \cdot g_0(x)
	\rightarrow \\\hat m(x) = \dfrac{\hat{c}_1(x) \mod g_0(x)}{g_1(x)}\\
	\rightarrow \hat m(x) = m(x)
	\end{array}
	\end{equation*}	
\end{proof}
\begin{remark}
We briefly comment on how to choose $r$ in Theorem~\ref{thm:cyclic_construction}.
To store $n-1$ information symbols, choose $r = 0 $ and $ g_1(x) = x^0 = 1 $ so that $\mycode{C}$ becomes $\mycode{C}_0$. 
For combined masking and error correction, choose $r$ between $1 \leq r< n-1$.
However, the chosen value of $r$ should be the smallest value such that an $[n$, $n-r$, $\delta_1]_q$ code exists. 
\end{remark}

\subsection{Further Decreasing the Redundancy in Construction~\ref{cons1}}
According to \cite[Construction~3]{wachterzeh2016codes}, it is possible to further decrease the required redundancy necessary for masking to be $<1$ in Construction~\ref{cons1} if the $v$ value is chosen from a small subset of $[q]$ such that $v \in [u+1]$ and $v \not \in {w}_{\phi_{i}} \mod (u+1) $. Since the set
$\{{w_\phi}_0 \mod (u+1), {w_\phi}_1 \mod (u+1), {w_\phi}_2\mod (u+1), \dots , {w_\phi}_{u-1}\mod (u+1)\}\mod (u+1)\} $ 
has the cardinality $u$ and there are $u+1$ possible values to chose from, we can always find $v \in [u+1]$. Thus, the stored information $k_1$ increases by $\log_q \lfloor \frac{q}{u+1} \rfloor$ which is the amount that the required redundancy for masking decreases $1- \log_q \lfloor \frac{q}{u+1} \rfloor$. Construction~\ref{cons1} will be modified such that $v$ and the coefficients of $g_0(x)$, $h_0(x)$, and $g_1(x) \in [u+1]$ (from a small subset of $q$), while $z_0\in [q]$ and $\ve{m} \in {\mathbb{F}_q^{k_1}}$. 
We show the improvement in Construction~\ref{cons1} for $k_1$ and $l$ by $k_1^*$ and $l^*$ columns in Table~\ref{table2}. We choose $q=6$ and $u+1 = 3$ to show the improvement for the same $n$ length.  

\subsection{Bounds for Partitioned Cyclic Code}
As aforementioned, we use a partitioned cyclic $\mycode{C}$ (involves $\mycode{C}_1$ and $\mycode{C}_0$) code to mask partially stuck memory cells and correct additional random errors. However, to derive bounds for the minimum distance for the error correction part of the code $\mycode{C}_1$, we apply the BCH bound. In the \emph{partitioned} cyclic  code, there is a pair of designed distances as $(\delta_1,$ $\delta_0)$, where $\delta_1$ is the designed distance of $\mycode{C}_1$ with parameters $[n,k_1+1,\delta_1]_q$ and
$\delta_0$ is the designed distance for the masking part of the code $\mycode{C}_0$ with parameters $[n,k_1+r,\delta_0]_q$. 

The lower bounds are the pair of designed distances $(\delta_1,$ $\delta_0)$ and they are computed as shown in the following:

	\begin{itemize}
		\item $\delta_0 = 2$, 
use a $[n, n-1,2]_2$ Single Parity Check Code.
\item $\delta_1$ follows always the chosen degree $r$ of $g_1(x)$ based on the given lower bound in \cite [Appendix]{heegard1983partitioned}:
		\begin{equation}\label{equar}
			r \leq m \left\lceil \frac{\delta_1 - 1}{2}\right\rceil.
		\end{equation}	
	\end{itemize}

In the following section, we present Table~\ref{table2} of length $n=8$ to compare ternary code for partially stuck memory cells to stuck cells with error correction as in \cite{heegard1983partitioned}.

\subsection{Table of Ternary Code with Code Length ($n = 8$)}
Table~\ref{table2} presents the maximum amount of information to be stored and the maximum errors that can be corrected. Comparing Table~\ref{table2} to Table~\ref{table3} from \cite{heegard1983partitioned}, an improvement is shown in Table~\ref{table2} since only one parity symbol for masking is used. 
For a ternary code $n= q^m-1=8$, for $m=2 \text{ and } q=3$, ($x^8-1$) factors are:\\ $(x+1) \cdot (x^2+2x+2) \cdot (x^2+1) \cdot (x+2) \cdot (x^2+x+2)  \equiv M^{(0)}(x)\cdot M^{(1)}(x)\cdot M^{(2)}(x) \cdot M^{(4)}(x)\cdot M^{(5)}(x)$ respectively. 

   Although the overall redundancy $r + l = 6$ used as shown in item~6 in Table~\ref{table2} and item~5 in Table~\ref{table3}, longer $r$ in Table~\ref{table2} increases the probability to correct more errors. This is an improvement from \cite{heegard1983partitioned}. Furthermore, Table~\ref{table2} shows more flexibility in storing $k_1$ symbols as shown in item 2 where $k_1=5$. For longer code length $n$, our construction is more likely to have longer consecutive cyclotomic cosets and accordingly higher minimum distance, thus, can correct more errors than \cite {heegard1983partitioned}.

 \section{New Codes for Partially Stuck-At Cells $s_{\phi_i} =1$, $q \leq u < n$} 

The masking technique in the previous section only guarantees successful masking up to a number of $u<q$ partially stuck-at-1 cells. In the following, we study the problem of masking at least $q$ cells. First, we determine the probability that masking is still possible with Construction~\ref{cons1} for random partially stuck positions. Based on \cite[Construction~4]{wachterzeh2016codes}, we propose a method to further increase the guaranteed number of masked cells at the cost of a masking redundancy of more than one symbol.

\subsection{Probabilistic Masking}
	
We determine the probability that for a random message, masking is possible for a number of stuck positions $u \geq q$ if the code constructions in Theorem~\ref{thm1} and Theorem~\ref{thm2} are used.
This \emph{probabilistic masking} approach enables us to use a row of memory with a certain probability even if more than $q-1$ partially stuck cells are present.

  \begin{thmmystyle}
  	\label{thm3}
	
	Let $\ve{G}$ be as in Theorem~\ref{thm1}, $q \leq u \leq k_1$, $s_{\phi_i} = 1$, and $\phi_1,\dots,\phi_u$ such that the columns of $\ve{G}$ indexed by the $\phi_i$ are linearly independent.
	For a message $\m \in \Fq^{k_1}$ that is drawn uniformly at random, the probability that we can mask the word is
	\begin{equation*}
  	\label{eq4}
		\mathrm{P}(q,u) = 
  		{\dfrac{ \sum_{i=1}^{q} (-1)^{i+1} {q \choose i}\cdot (q-i)^{u}}{q^{u}}}
	\end{equation*}
  \end{thmmystyle} 
  
	\begin{proof}
	The code discussed in Theorem~\ref{thm1} guarantees to mask $u<q$ partially stuck cells since $u$ values at the stuck positions of the intermediate codeword $\ve{w}$ (after encoding only the message) do not cover the entire alphabet $\{0,\dots,q-1\}$. Hence, we can add a constant to all of them in order to prevent the codeword values to be $0$ in the stuck positions.
	For the probabilistic case, we can apply the same arguments, but we need to derive a probability that the intermediate codeword values at the partially stuck positions do not constitute the entire alphabet.

	Due to the assumption on the linear independence of $\ve{G}$'s columns, the vector $(w_{\phi_1},\dots,w_{\phi_u})$ is uniformly distributed on $\Fq^u$. The formula for $\mathrm{P}(q,u)$ uses the inclusion-exclusion principle~\cite{kharazishvili2008combinatorial} to count the relative number of vectors in $\Fq^u$ that exclude at least one field element.
	\end{proof}
\begin{table*}[h]
	\caption{Comparison between \cite {heegard1983partitioned}, \cite {wachterzeh2016codes}, and this work. Notation: See Section~\ref{ssec:notation}.}
	\label{table1}
	\vspace{-.4 cm}
\begin{center}
\scalebox{1}{
\def\arraystretch{1.5}
\begin{tabular}{ |p{2cm}| p{5cm} || p{3 cm} | | p{6 cm} |}
	\hline
	&\textbf{\!\!\cite{heegard1983partitioned} (Stuck cells)} & \textbf{\!\!\cite{wachterzeh2016codes} (Partially stuck cells)} & \textbf{This Work (Partially stuck cells)} \\ \hline
	error correction & yes & no & yes \\ \hline
	 $k_1$&If $l = 1$, then $k_1 = n- r  - 1$& If $l = 1$, then $k_1$ $ \leq n - 1$ &  $k_1$ $ \leq n - 1 -r$. If $r =0$, it is similar to \cite {wachterzeh2016codes} \\ \hline
	 $l$& $l = n - k_1 - r$&  $l = n - k_1 = 1$ &  $l = n - k_1 -r = 1$\\ \hline
	 $r$& $r = n - k_1 - l$ & None &$r = n - k_1 - 1$ \\ \hline
	$\delta_0$ and $\delta_1$ & $\delta_0$ for masking and $\delta_1$ for error correction  & $\delta_0 = 2$ \vphantom{\LARGE H} & $\delta_0 = 2$ , $\delta_1$ follows the chosen $r \rightarrow$ $r \leq m \lceil \tfrac{\delta_1 - 1}{2}\rceil$ \\ \hline
	$t$& $\leq \lfloor\frac{\delta_1 -1 }{2}\rfloor$ & $0$ &$\leq \lfloor\frac{\delta_1 -1 }{2}\rfloor$ \\ \hline
	$u$& $u < \delta_0$ and $2t < \delta_1$, Or $u \geq \delta_0$ and $2(u+t+1-\delta_0) < \delta_1 $& $u \leq n$ and $u<q$&  $u\leq \min \{n,q-1\}$ (Theorem~\ref{thm1} and~\ref{thm2}), or $u\leq n$ (Theorem~\ref{thm3} and~\ref{thm4}) \\ \hline
	$s_{\phi_i}$&All levels  ($0,\dots,q-1$)  & Partial levels ($1, \dots, q-1$)& $1$ but expendable to arbitrary $s_{\phi_i}$  \\ \hline
	Name& $u$-SMC  & $u$-PSMC &($u$, $t$)-PSMC\\ \hline

\end{tabular}
}
\end{center}
\end{table*}
The following example illustrates that the probability that masking is successful can be quite large.

 \emph{Example}. Let $q =3$, $n=8$, $r =0$ and $v \in [q]$, the probability in which we can mask $u=n-1$ partially stuck-at-1 memory cells because the values at the partially stuck positions are from a set of size at most $q-1$ is:
  	\begin{equation*}  
     		{\frac{3\cdot(3 -1)^{7} - {3 \choose 2} \cdot (3 -2)^{7} + {3 \choose 3} \cdot (3 -3)^{7}}{3^{7}}}= 0.17. 
  \end{equation*}
  This ratio will be $0.77$ if $u=q$ and clearly it is $1$ if $u<q$. Let $\ve{m} = (2002220)\in \F_{3}^{7}$ so the augmented message vector $\ve{w} = (\textcolor{orange}{0}2002220)\in \F_{3}^{8}$. Let the partially stuck positions be ${w_\phi}_0, {w_\phi}_1, {w_\phi}_2, \dots , {w_\phi}_i$ for all $i \in [u]$ and $u=n-1$. Thus, $v = 1$ and $z_0 = 2$. The output from Algorithm~\ref{tab:a1} in Theorem~\ref{thm1} is:
  \begin{equation*}
  \ve{c} = (02002220) + (22222222) = (21221112).
  \end{equation*}
  As it is shown from the output vector, the partially stuck-at-1 positions for all $u$ are masked because they fulfill the following condition:
  \begin{equation*}
  {w_\phi}_i + z_0 \mod q  \geq 1,  \text{ where } q \leq u < n. 
  \end{equation*} 

\begin{remark}
The assumption in Theorem~\ref{thm3} that the columns of $\ve{G}$ indexed by the partially stuck positions are linearly independent is fulfilled for most known codes with high probability if $u\leq k_1$, especially if $u \ll k_1$. For dependent columns, it becomes harder to count the number of vectors that cover not the entire alphabet since $(w_{\phi_1},\dots,w_{\phi_u})$ is uniformly distributed on a subspace of $\Fq^{k_1}$.
\end{remark}
    	
\subsection{Masking Up to $u\leq q+d_0-3$ Partially Stuck-At Cells}

Construction~4 in \cite{wachterzeh2016codes} masks $q \leq u \leq n$ partially stuck-at-1 cells. It is a generalization of the all-one vector as stated \cite[Theorem~4]{wachterzeh2016codes} because if $d_0 = 2$, $u\leq q-1$. Theorem~\ref{thm1} in this paper provides a construction to mask and correct errors using a matrix with specific form that has ($\vec{G}_1$ and $\vec{G}_0$). Extending Theorem~\ref{thm1} by using the parity-check matrix in \cite[Construction~4]{wachterzeh2016codes} instead of $\ve{G}_0$ given in Theorem~\ref{thm1} will provide a solution to mask $u\leq q+d_0-3$ and correct $t$ errors. The result is formulated in Theorem~\ref{thm4}.    
  \begin{thmmystyle}
	\label{thm4}
	Let $u\leq q+d_0-3$ and $s_{\phi_i} = 1$. Assume there is an $[n, k, d_0 \geq u-q+3,d_1 \geq 2t+1]_q$ code $\mycode{C}$ with  a $k \times n$ generator matrix of the following form:
	\begin{center}			
	\scalebox{1}{	
		$ \ve{G} = \begin{bmatrix}  \ve{G}_1\\ \ve{H}_0  \end{bmatrix} = \begin{bmatrix} \ve{0}_{k_1 \times l} & \ve{I}_{k_1} & &&\ve{P}_{k_1\times r} \\ &\ve{H}_0   &   \end{bmatrix}$}
		
	\end{center}
	where: 
	\begin{itemize}
		\item $\ve{G}_1$ is a $k_1 \times n$ generator matrix of an $[n, k_1]_q$ code $\mycode{C}_1$.
		\item For simplicity $\ve{H}_0$ is a systematic $l \times n$ parity-check matrix of an $[n, k_1+r, d_0]_q$ code $\mycode{C}_0$.  
		\item $k_1 = n - l - r$.
		\item $k = k_1 + l$.
		\item $\ve{P} \in \mathbb{F}^{k_1\times r}_q $.
	\end{itemize}
	
	By using Algorithm~\ref{a5} and ~\ref{a6}, this code is a ($u$,$t$)-PSMC which can mask any $u\leq n$ partially stuck memory cells, and can correct $t$ additional random errors while storing $k_1 = n-r-l$ information symbols. 		
\end{thmmystyle} 
\begin{algorithm}
	
	\caption{Encoding}
	\label{a5}
		\KwIn{
	\begin{itemize}
		\item Message: 
		$m = (m_0, m_1,\dots, m_{k-1})  \in {\mathbb{F}_q^{k_1}}$
		\item Positions of partially stuck cells: $\ve{\phi}$
	\end{itemize}
}
	
		 Compute the final message vector $\ve{w} = \ve{m} \cdot \ve{G}_1$
		 
		 Find  $\ve{z} = (z_0,z_1,\cdots z_{l-1})$ $\in \mathbb{F}_q$ as explained in the proof of \cite[Theorem~7]{wachterzeh2016codes} 
		 	
		 Compute the masking vector $\ve{d} = \ve{z} \cdot \ve{H}_0$
		 
		 Compute $\ve{c} = (\ve{w} + \ve{d})  \mod q$

		\KwOut{
	Codeword $\ve{c} \in \mathbb{F}^n_q$}

\end{algorithm}
\vspace{-0.3cm}
\begin{algorithm}
	\label{a6}
	\caption{Decoding}	
	\KwIn{$\ve{y} = \ve{c}+ \ve{e} \in \mathbb{F}^n_q $}

	$\hat{\c} \gets$ decode $\y$ in the code $\mathcal{C}$
		
		Unmasking Process:\\
		\begin{itemize}
			\item $\hat{\ve{z}}\leftarrow (\hat{c}_0,\hat{c}_1,\cdots \hat{c}_{l-1})$
			
			\item $\hat {\ve{w}} = (\hat{w_0}, \hat{w_1},\cdots, \hat{w}_{n-1}) \leftarrow (\hat{\c} - \hat{\ve{z}} \cdot \ve{H}_0)  \mod q$
			
			\item $\hat{\ve{m}}  \leftarrow (\hat{w}_l,\hat{w}_{l+1} \dots, \hat{w}_{n-r-1})$
		\end{itemize}

	\KwOut{
	Message  vector $\hat{\ve{m}}  \in \mathbb{F}^{k_1}_q$  }
	
\end{algorithm}	
\vspace{-0.5cm}

\begin{proof}
	Algorithm~\ref{a5} finds $\ve{z}$ similar to \cite[Algorithm~7]{wachterzeh2016codes} instead of only finding $v$ value as shown in Algorithm~\ref{tab:a1}. 
	The proof follows a detailed proof in \cite[Theorem~8]{wachterzeh2016codes} for the masking part which replaces the all-ones vector with an $(n-k)\times n$ parity-check matrix. However, Theorem~\ref{thm4} uses $\ve{H}_0$ which is $(n-k_1-r)\times n$.  \\
	Since the error correction is related to $[n, k_1]_q$ code $\mycode{C}_1$ and matrix $\ve{G}_1$, the proof follows Theorem~\ref{thm1} for how to choose $r$ and accordingly correct $t$ errors. 
\end{proof}

Theorem~\ref{thm4} provides an extended construction of Construction~\ref{cons1} to mask any $u\leq n$ partially stuck-at-1 cells and correct $t$ errors.
This gain in the number of partially stuck cells that can be masked comes at the cost of larger redundancy. However, the redundancy is still smaller than the construction for masking and error correction in \cite{heegard1983partitioned} for the same number of (in this case stuck) cells.

\section{Conclusion}

We have proposed several constructions for combined masking of partially stuck-at-1 cells and error correction, by combining the methods proposed in \cite{heegard1983partitioned} for error correction and masking stuck cells, with the masking-only codes for partially stuck cells in \cite{wachterzeh2016codes}.
Compared to \cite{wachterzeh2016codes}, the new code constructions can correct errors in addition to masking. Furthermore, less redundancy is required for masking compared to the code constructions for stuck cells in \cite{heegard1983partitioned}.

The results can be extended to partially stuck-at-$s$ cells for $s>1$ using similar methods as \cite[Section~VII]{wachterzeh2016codes}.
It is also possible to extend the cyclic construction in \cite[Theorem~2]{heegard1983partitioned} using the approach in Theorem~\ref{thm4}, as well as the construction based on binary codes in \cite[Construction~5]{wachterzeh2016codes}.
Furthermore, future work should derive bounds on the required redundancy for a given number of partially stuck cells to mask and number of errors to correct.

\bibliographystyle{IEEEtran}
\bibliography{main}

\section {Appendix}
\emph{Example. Masking and Correcting Partially Stuck Memory Cells - Matrix Form } 
Let  $q = 3$ and $n = 14$ and we want to store the message vector $\ve {m}_1 = (0210210210210)$  $\in \mathbb{F}_q^{n-1}$ or $\ve{m}_2 = (0210210210) \in \mathbb{F}_q^{n-1-r}$ and correct one error $t=1$ in the tenth position. We use a ternary Hamming code with redundancy $r=3$. The partially stuck positions named $\phi_i$ are $\phi_1 =4$ and $\phi_2 = 6$, $\forall i \in u$ and $ u < q$ and $u$ is the number of partially stuck cells. According to Algorithm~\ref{tab:a1}, we need to find the following:\\

	\scalebox{0.9}{$\ve{G} = 	\left[
	\begin{array}{*{1}c} {\ve{G}_1} \\ \ve{G}_0 \end{array}
		\right] =  \begin{bmatrix} \ve{0}_{13 \times 1} & \ve{I}_{13} \\  {1} & {\ve{1}_{1 \times 13}}  \end{bmatrix} $}, and 	$ \ve{P}_{(\textcolor{red}{10}\times \textcolor{red}{3})} = $	\scalebox{.6}{	$	\left[
		\begin{array}{*{3}c}
		
		\textcolor{red}{1}& \textcolor{red}{2}&\textcolor{red}{0}\\
		\textcolor{red}{0}& \textcolor{red}{1}& \textcolor{red}{2}\\
		\textcolor{red}{1}& \textcolor{red}{0}& \textcolor{red}{2}\\
		\textcolor{red}{1} &\textcolor{red}{1} &\textcolor{red}{1}\\
		\textcolor{red}{1} &\textcolor{red}{1}& \textcolor{red}{2}\\
		\textcolor{red}{2}& \textcolor{red}{0}& \textcolor{red}{2}\\
		\textcolor{red}{1} &\textcolor{red}{2} &\textcolor{red}{1}\\
		\textcolor{red}{2} &\textcolor{red}{1} &\textcolor{red}{1}\\
		\textcolor{red}{2} &\textcolor{red}{2} &\textcolor{red}{0}\\
		\textcolor{red}{0} &\textcolor{red}{1} &\textcolor{red}{1}\\

	\end{array}
\right]$}

\begin{enumerate}
		
	\item \emph{Error free} where $r = 0$, find $\ve{w}_1 = \ve {m}_1 \cdot \ve{G}_1$.
	\item \emph{With error $t=1$, find $\ve{w}_2$ },\\
	\scalebox{0.9}{	$\ve{G} = \left[
		\begin{array}{*{1}c} \ve{G_1}\\ \ve{G}_0 \end{array}	
		\right] = \begin{bmatrix} \ve{0}_{10 \times 1} & \ve{I}_{10} & \ve{P}_{10\times 3} \\  {1} & {\ve{1}_{1 \times 10}} & {\ve{1}_{10 \times 3}} \end{bmatrix}$},\\ 
	
	then $\mathbf {w}_2 = \mathbf {m}_2 \cdot \mathbf{G}_1 $.
		
		\item Find $\mathbf{d}$. Since the partially stuck positions are $w_{\phi_1} = 0$ and $w_{\phi_2} = 1$, $v \neq w_{\phi_1} \text{ and }  \neq w_{\phi_2}\text{, then } v = 2$. 	
		Thus, $z_0 = 1$.
		\begin{center}
	\scalebox{1}{	$\ve{d} = z_0 \cdot \ve{G}_0 \rightarrow \ve{d} = 1 \cdot \ve{1}_{1 \times 14}$}
\end{center} 
    \item The codeword vector that can mask only is: \begin{center} $\ve{c} = \ve{w}_1 + \ve{d}$ 
		
	 $\rightarrow \ve{c} = \scalebox{1}{1102\textcolor{red}{1}0\textcolor{red}{2}1021021}$.	\end{center}  
	\item The codeword vector that can mask and correct is:	\begin{center} $\ve{c} = \ve{w}_2 + \ve{d}$
 $\rightarrow \ve{c} = \scalebox{1}{1102\textcolor{red}{1}0\textcolor{red}{2}1021000}$ \end{center} 
	\item \label{it6} Compute $	\ve{s} = \ve{y} \cdot \ve{H}^T$ where $\ve{H}$ is the parity check matrix of $\ve{G}$. If $\ve{s}= \ve{0}$, the retrieved vector $\ve{y}$ is error free,  $z_0 \leftarrow c_0$:	
$ 	\hat {\ve {w}}_1 = \ve{c} -  1 \cdot \mathbf{1}_{1 \times 14}$	
$= \scalebox{1}{\textcolor{orange}{0}0210210210210}$. \\Then, $\hat{\ve{m}}  \leftarrow (\hat{w}_1, \dots, \hat{w}_{13})$ and we get $\hat{\ve{m}} = \ve{m}$.
	\item To decode the retrieved vector $\ve{y}$ with single error $t =1$ at the tenth position, $\ve{y}= \ve{c} + \ve{e} =  \scalebox{1}{1102\textcolor{red}{1}0\textcolor{red}{2}10\textcolor{blue}{0}1000}$.	
Compute $	\ve{s} = \ve{y} \cdot \ve{H}^T= 110$.	
It is ternary hamming code. Comparing $\ve{s}^T$ to the tenth column of $\ve{H}$ is to correct the error.
For unmasking, same as before (repeat Step~\ref{it6}).

\end{enumerate}

\end{document}